\documentclass[11pt]{article}

\usepackage[a4paper,textwidth=16cm,textheight=24cm]{geometry}

\usepackage{amssymb}
\usepackage{amsthm}
\usepackage{amsmath}
\usepackage{amsfonts}
\usepackage{multirow}

\usepackage{graphicx}
\usepackage{natbib}
\usepackage[dvipsnames]{xcolor}
\usepackage{authblk}

\usepackage{todonotes}
\usepackage{comment}
\includecomment{coment}

\usepackage{tikz}
\usetikzlibrary{arrows,intersections,fadings,calc}

\newtheorem{theorem}{Theorem}[section]
\newtheorem{lemma}[theorem]{Lemma}
\newtheorem{proposition}[theorem]{Proposition}
\newtheorem{corollary}[theorem]{Corollary}

\theoremstyle{definition}

\newtheorem{definition}[theorem]{Definition}
\newtheorem{example}[theorem]{Example}

\newcommand{\R}{\mathbb{R}}
\newcommand{\F}{\mathbb{F}}

\newcommand{\E}{\mathbb{E}}
\newcommand{\iid}{{\rm i.i.d.}}
\newcommand{\wrt}{{\rm w.r.t.}}
\newcommand{\eqd}{\stackrel{d}{=}}

\newcommand*{\TickSize}{2pt}%

\title{Inequalities and bounds for expected order statistics from transform-ordered families}
\date{}

\author{Idir Arab}
\affil{CMUC, Dep. Mathematics, Univ. Coimbra, Portugal; idir.bhh@gmail.com}

\author{Tommaso Lando}
\affil{Department of Economics, University of Bergamo, Italy; tommaso.lando@unibg.it}

\author{Paulo Eduardo Oliveira}
\affil{CMUC, Dep. Mathematics, Univ. Coimbra, Portugal; paulo@mat.uc.pt}

\begin{document}

\maketitle

\begin{abstract}
We introduce a comprehensive method for establishing stochastic orders among order statistics in the \iid\ case. This approach relies on the assumption that the underlying distribution is linked to a reference distribution through a transform order. Notably, this method exhibits broad applicability, particularly since several well-known nonparametric distribution families can be defined using relevant transform orders, including the convex and the star transform orders.
Moreover, for convex-ordered families, we show that an application of Jensen's inequality gives bounds for the probability that a random variable exceeds the expected value of its corresponding order statistic.

\smallskip


\end{abstract}



\section{Introduction} 
Order statistics are fundamental tools in probability, statistics, and reliability theory. Especially in the context of reliability, a major issue consists of comparing order statistics with different ranks and sample sizes. To be more specific, let $X$ be a random variable (RV) and denote with $X_{k:n}$ the $k$-th order statistic corresponding to an \iid\ random sample of size $n$ from $X$. If $X$ represents the lifetime of some component, then $X_{k:n}$ is the lifetime of a $k$-out-of-$n$ system, that is, a system that fails if and only if at least $k$ components stop functioning.
The ageing and reliability properties of such systems, described in terms of their stochastic behaviour, are an important aspect. Hence, the issue of comparing, in some stochastic sense, the order statistics $X_{i:n}$ and $X_{j:m}$, corresponding to systems with a different number of components and different functioning requirements, naturally arises. This problem can be addressed by the theory of \textit{stochastic orders} (see~\cite{shaked2007} for general results and relationships). In particular, several results on the stochastic comparison between order statistics have been obtained, for example, by~\cite{arnold1991,arnold1991exp,kochar2006,kochar2009,kochar2012,lando2021st}.

This paper focuses on establishing conditions under which $X_{i:n}$ dominates $X_{j:m}$ in the sense that $\E u(X_{i:n})\geq \E u(X_{j:m})$ for every function $u$ in some class $\mathcal{U}$. Relationships of this kind are referred to as \textit{integral stochastic orders} with respect to a generator class $\mathcal{U}$, as defined by~\cite{muller1997}. These orders include comparisons of expected order statistics when $\mathcal{U}$ contains the identity function. Significant examples of integral stochastic orders include the \textit{increasing concave} (ICV), \textit{increasing convex} (ICX), and the \textit{star-shaped} (SS) orders (see~{\cite{shaked2007}}).
In contrast to numerous methods found in the literature~\cite{arnold1991exp,wilfling1996c,kundu2016}, our approach does not assume a known parametric form for {the cumulative distribution function (CDF) $F$ of the RV $X$}.
Instead, we opt for a more flexible approach, by making nonparametric assumptions about $F$. Specifically, we assume that $G^{-1} \circ F\in\mathcal{H}$, where $G$ is a carefully chosen cumulative distribution function, and $\mathcal{H}$ represents a set of increasing functions. In other words, we assume that $F$ is related to $G$ via a \textit{transform order}~{\cite{sjs}}, where $\mathcal{H}$ is referred to as the generator class.
Interesting examples of distributions satisfying transform order assumptions are the \textit{increasing hazard rate} (IHR), \textit{increasing hazard rate average} (IHRA), \textit{increasing odds rate} (IOR), \textit{decreasing density} (DD), \textit{decreasing density on average} (DDA), \textit{decreasing reversed hazard rate} (DRHR) families (see~{\cite{shaked2007,marshall2007,oddspaper}}).
{In this paper, we show} that a key step for deriving {appealing probabilistic} inequalities between order statistics within transform-ordered families involves combining integral and transform orders with the same generator class. Additionally, we illustrate the application of this approach by deriving bounds for expected {values of} order statistics.
Our method's general behaviour aligns with expectations: stronger assumptions on $F$ lead to more applicable ordering conditions between {$X_{i:n}$} and $X_{j:m}$, or more stringent bounds, and vice versa.

The paper is organized as follows. In Section~\ref{sec:gen}, we present formal definitions and outline our general approach. Although our result is of general form, its application extends seamlessly to well-known classes of distributions, discussed in Section~\ref{sec:classes}. Section~\ref{sec:conv} delves into the derivation of conditions for the ICV and ICX order between order statistics from convex-ordered families, extending some recent results of~\cite{lando2021st}.
Moving on to Section~\ref{sec:star}, we establish conditions for the SS order between order statistics within star-ordered families. The approach is heuristically extended to \textit{increasing anti-star-shaped} (IAS) order, introduced in Section~\ref{sec:gen}, based on a simulation algorithm. Finally, Section~\ref{sec:bounds} provides bounds for the probability that $X$ exceeds its expected order statistic $\E X_{i:n}$, that is, the probability that a single component surpasses the expected lifetime of the system. As a byproduct of this general result, we provide a new characterisation of the log-logistic distribution (with shape parameter 1).

\section{Preliminaries and some definitions}
\label{sec:gen}
Throughout this paper, ``increasing'' and ``decreasing'' are taken as ``non-decreasing'' and ``non-increasing'', respectively, and the generalised inverse of an increasing function $v$ is denoted as $v^{-1}(u)=\sup\{x\in\mathbb{R}:\,v(x)\leq u\}$.
Moreover, the beta function with parameters $a,b>0$ is denoted with $\mathcal{B}(a,b)=\int_0^1 t^{a-1}(1-t)^{b-1}\,dt$.
Finally, given an absolutely continuous CDF, its density function is denoted by the corresponding lowercase letter.

We shall consider two general families of stochastic orders, characterised either by integration or shape assumptions, which are shown to be crucial for establishing comparisons between expectations of order statistics.
\begin{definition}[\cite{muller1997}]\label{i order}
Let $\mathcal{U}$ be some family of functions. We say that $X$ dominates $Y$ in the $\mathcal{U}$-integral stochastic order, denoted as $X\geq_\mathcal{U}^I Y$, if {$\E u(X)\geq\E u(Y)$} for every $u\in \mathcal{U}$, provided that the integrals exist. $\mathcal{U}$ is referred to as the generator of the integral order.
\end{definition}
Setting particular choices for the generator class we obtain some well known stochastic orders. We recall some relevant classes of functions before presenting the translation of the previous definition into specific ordering relations.
\begin{definition}
A non-negative function $h(x)$, defined for {$x\geq 0$} and such that $h(0)=0$, is said to be
    \begin{enumerate}
    \item
    star-shaped at the origin if every segment joining the origin with the graph of $h$ {always stays} above the graph, or, equivalently, if $\frac {h(x)}x$ is increasing;
    \item
    anti-star-shaped (at the origin) if every segment that joins the origin with the graph of $h$ is always below the graph, or, equivalently, if $\frac {h(x)}x$ is decreasing.
    \end{enumerate}
\end{definition}

We will focus on the following integral stochastic orders, obtained from Definition~\ref{i order} for particular generator classes.
\begin{definition}
Assume that $X\geq_\mathcal{U}^IY$. We say that $X$ dominates $Y$ in
\begin{enumerate}
\item
the usual stochastic order, denoted as $X\geq_{st}Y$, if $\mathcal{U}$ is the family of increasing functions;
\item
the increasing concave (ICV) order, denoted as $X\geq_{icv}Y$, if $\mathcal{U}$ is the family of increasing concave functions;
\item
the increasing convex (ICX) order, denoted as $X\geq_{icx}Y$,  if $\mathcal{U}$ is the family of increasing convex functions;
\item
the star-shaped (SS) order, denoted as $X\geq_{ss}Y$, if $\mathcal{U}$ is the family of star-shaped functions;
\item
the increasing anti-star-shaped (IAS) order, denoted as $X\geq_{ias}Y$, if $\mathcal{U}$ is the family of increasing anti-star-shaped functions.
\end{enumerate}
\end{definition}
The ICV, ICX, and SS orders are well known (see for instance \cite{shaked2007}). Differently, the IAS order seems not to have been studied.
As we will discuss in Section~\ref{sec:ias.ss}, the IAS order has the disadvantage, unlike the others, that an easy to check characterization is not available.
The relationships among classes of functions yield the following implications (see Theorem~4.A.55 in \cite{shaked2007} for the first line, while the second is proved later in Proposition~\ref{prop:ias}):
$$
\begin{array}{ccccc}
X\geq_{st}Y & \Longrightarrow & X\geq_{ss}Y & \Longrightarrow & X\geq_{icx}Y \\
X\geq_{st}Y & \Longrightarrow & X\geq_{ias}Y & \Longrightarrow & X\geq_{icv}Y
\end{array}
$$

All these orders imply inequality of the means, $\E X\geq \E Y$ since the identity function belongs to each of the above classes.

We now introduce a second general family of stochastic orders.
\begin{definition}[\cite{sjs}]\label{t order}
Let {$\mathcal{H}$} be some family of increasing functions. We say that $X\sim F$ dominates $Y\sim G$ in the {$\mathcal{H}$}-transform order, denoted as {$X\geq_\mathcal{H}^T Y$}, or, equivalently, {$F\geq_\mathcal{H}^T G$}, if $F^{-1}\circ G\in{\mathcal{H}}$. {$\mathcal{H}$ is referred to as the generator of the transform order $\geq_\mathcal{H}^T$.}
\end{definition}
Similarly to the integral stochastic orders defined earlier, the following transform orders may be obtained from Definition~\ref{t order} by taking $\mathcal{H}$ as the class of convex and star-shaped functions, respectively.
\begin{definition}
Assume that $X{\geq_\mathcal{H}^T}Y$. We say that $X$ dominates $Y$ in
\begin{enumerate}
\item
the convex transform order, denoted as $X\geq_c Y$, if ${\mathcal{H}}$ is the family of (increasing) convex functions;
\item
the star order, denoted as $X\geq_* Y$, if ${\mathcal{H}}$ is the family of star-shaped functions.
\end{enumerate}
\end{definition}

We should note that the standard stochastic order may be seen {both} as an integral and a transform order. In fact, $X\geq_{st}Y$ if $F^{-1}\circ G(x)\leq x$, for every $x$.

In this article, we show that a useful approach to obtaining interesting stochastic inequalities consists in a suitable combination of integral and transform orderings based on a common generator class.

\section{Main result}
\label{sec:main}

We now address the comparison of order statistics with respect to integral stochastic orders. We shall be taking $F$ as the CDF of interest, and $G$ some suitably chosen reference CDF. It is well known that the CDF of $X_{i:n}$ is given by $F_{B_{i:n}} \circ F$, where $F_{B_{i:n}}$ is the CDF of a beta random variable with parameters $i$ and $n-i+1$, that is, $B_{i:n}\sim beta(i,n-i+1)$ (see \cite{jones2004}).
Equivalently, one can write $X_{i:n}{\eqd} F^{-1}\circ B_{i:n}$. This representation renders it difficult to establish conditions for a stochastic comparison between two different order statistics, say {$X_{i:n}$} and $X_{j:m}$, since the result depends on the four parameters $i,j,n,m$ and on the analytical form of $F$. In a parametric framework, $F$ is assumed to be known {up to defining several real parameters}, so the problem boils down to a mathematical exercise, which may still be analytically complicated.
However, if $F$ is {in some nonparametric class}, the problem is more complicated, and{, as we show in the sequel,} it can be solved just by adding some {shape} constraints on $F$. {In this nonparametric framework}, results may still be obtained by applying a simple decomposition trick{: write $X_{i:n}\eqd F^{-1}\circ G\circ G^{-1}\circ B_{i,n}$ and assume that $F$ is related to some known $G$ by a suitable transform order}. Indeed, in this case, the analytical form of $G$ {being} known, the problem reduces to a simpler comparison between known RVs, namely $G^{-1}\circ B_{i:n}$ and $G^{-1}\circ B_{j:m}$.

For the sake of convenience and flexibility in the applications of our main result, we introduce the following notation.
\begin{definition}
\label{dom_class}
Let $G$ be some CDF and $\mathcal{H}$ some family of increasing functions. We define $\mathcal{F}_\mathcal{H}^G=\{F:F\geq_\mathcal{H}^TG\}$, that is, the family of CDFs that dominate $G$ with respect to the $\mathcal{H}$-transform order.
\end{definition}
{We may now state our main result}, which establishes sufficient conditions for comparing {expected} order statistics.
\begin{theorem}\label{general}
Let $\mathcal{H}$ be a class of increasing functions. If, for some given CDF $G$, $X\sim F\in \mathcal{F}_\mathcal{H}^G$, {$G^{-1}\circ B_{i:n}\geq_\mathcal{H}^I G^{-1}\circ B_{j:m}$} and $\geq_\mathcal{H}^I$ is preserved under $\mathcal{H}$ transformations, then $X_{i:n}\geq_{\mathcal{H}}^I X_{j:m}$.
\end{theorem}
\begin{proof}
Writing {$X_{i:n}\eqd F^{-1}\circ G\circ G^{-1}\circ B_{i:n}$}, the result follows easily from the definitions above. In fact, the order $\geq_\mathcal{H}^I$ is preserved under $\mathcal{H}$-transformations, whereas the assumption $F\in \mathcal{F}_\mathcal{H}^G$ ensures that $F^{-1}\circ G$ is an $\mathcal{H}$-transformation. Therefore, applying the transformation $F^{-1}\circ G$ to both sides of the stochastic inequality {$G^{-1}\circ B_{i:n}\geq_\mathcal{H}^I G^{-1}\circ B_{j:m}$}, we obtain {$X_{i:n}\geq_\mathcal{H}^I X_{j:m}$}, which implies the desired result by definition of integral stochastic orders, {taking into account that} $F^{-1}\circ G\in\mathcal{H}$.
\end{proof}
Note that, if $\mathcal{H}$ is closed under the composition of functions, the preservation assumption of the $\geq_\mathcal{H}^I$ order is automatically fulfilled. Despite the simplicity of Theorem~\ref{general}, its applications are remarkably
interesting, showcasing the profound implications of the interplay between integral and transform orders.

\section{Types of class generators}
\label{sec:classes}
Definitions~\ref{i order} and \ref{t order} become particularly interesting when the generator classes are chosen as well-known and popular families. We will now show that some of the already mentioned classes are encompassed within this framework, and add a number of further interesting families of distributions that can also be addressed. Indeed, according to the choice of the class $\mathcal{H}$ and of the reference CDF $G$ in Theorem~\ref{general}, these choices yield different families of the type $\mathcal{F}_\mathcal{H}^G$, which, we recall, are defined via a transform order.
As shown below, when $\mathcal{H}$ is the class of increasing convex or concave functions, $\mathcal{F}_\mathcal{H}^G$ may be characterised using the convex transform order. Hence, we will refer to these choices of $\mathcal{H}$ as \textit{convex-ordered} families.
Similarly, when $\mathcal{H}$ is the class of star-shaped or (increasing) anti-star-shaped functions, $\mathcal{F}_\mathcal{H}^G$ may be characterised via the star transform order, so we will {refer to these choices} as \textit{star-ordered} families.
For the sake of simplicity, besides the already defined classes $\mathcal{C}$ of convex functions, and $\mathcal{S}$ of functions that are star-shaped at the origin, we shall define $\mathcal{V}$ as the class of concave functions, and $\mathcal{A}$ as the class of increasing anti-star-shaped functions. Bear in mind that a function is convex if and only if its inverse is concave, so that $F\geq_\mathcal{C}^TG$ is equivalent to $F\leq_\mathcal{V}^TG$. The same relation holds between star-shaped and increasing anti-star-shaped functions, namely, $F\geq_\mathcal{S}^TG$ is equivalent to $F\leq_\mathcal{A}^TG$. This is stated as follows.
\begin{lemma}
\label{ias-ss-1}
$h$ is star-shaped if and only if $h^{-1}$ is increasing anti-star-shaped.
\end{lemma}
\begin{proof}
Let $h$ be star-shaped, so $\frac{h(x)}{x}$ is increasing. Note that $h$ is strictly increasing by construction, but it may have jumps, corresponding to intervals at which the generalised inverse $h^{-1}$ is constant. Proceeding by composition, $\frac{y}{h^{-1}(y)}$ is increasing, even in those intervals where $h^{-1}$ is constant. So the ratio $\frac{h^{-1}(y)}{y}$ is decreasing, concluding the proof.
\end{proof}

The results that follow from Theorem~\ref{general} depend obviously on the choice of $G$. In particular, we will consider the uniform distribution on the unit interval, with CDF $U(x)=x$, $x\in[0,1]$, the exponential distribution, with CDF $\mathcal{E}(x)=1-e^{-x}$, $x\geq0$, the standard logistic distribution with CDF ${L}(x)=\frac{1}{e^{-x}+1}$, $x\in\R$, {and} the log-logistic distribution with shape parameter equal to 1, hereafter LL1, with CDF ${LL}(x)=\frac x{1+x}$, $x\geq0$. These reference distributions, as described below, lead to several well-known families of distributions.
We will also consider the corresponding ``\textit{negative}'' versions: in general, if $Y\sim G$ then $-Y\sim G_-$, where $G_-(x)=1-G(-x)$. Note that, due to symmetry, for the logistic distribution we have $L=L_-$.

Combining the classes $\mathcal{C}$, $\mathcal{V}$, $\mathcal{S}$, or $\mathcal{A}$ with the choices of $G$ discussed above, we may generate {several different} families of distributions{, some of them well known in the literature. An application of Theorem~\ref{general} will derive inequalities that hold for each of the constructed classes of distributions}. Naturally, some of these are more interesting than others. Hereafter we will focus on the following ones.
\begin{enumerate}
\item
The class of concave CDFs, also known as \textit{decreasing density} (DD) class, as it requires the existence of a decreasing PDF (except, possibly, at the right-endpoint of its support). This may be obtained by $\mathcal{F}_{\mathcal{C}}^{U}=\{F:F\geq^T_{\mathcal{C}}U\}=\{F:F\geq_{c}U\}=\{ F^{-1}\in\mathcal{C}\}={\mathcal{V}}$. This class has received much attention in the literature, for instance, it is a typical assumption for shape-constrained statistical inference (see, for example,~\cite{groeneboom2014}). Among known parametric models, the gamma, the log-logistic, and the Weibull distributions, with shape parameters less than or equal to 1, belong to this class.

\item
The class of convex CDFs, also known as \textit{increasing density} (ID) class, as it requires the existence of an increasing PDF (except, possibly, at the right-endpoint of its support). This may be obtained by $\mathcal{F}_{\mathcal{V}}^{U}=\{F:F\geq^T_{\mathcal{V}}U\}=\{F:U\geq_{c}F\}=\{ F^{-1}\in\mathcal{V}\}={\mathcal{C}}$. This class is generally less applicable than the DD one, as it requires bounded support, and contains few known parametric models.

\item
The class of star-shaped CDFs. In the case of absolutely continuous distributions, this is also known as the class of distribution with \textit{increasing density on average} (IDA). This may be obtained as $\mathcal{F}_{\mathcal{A}}^{U}=\{F:F\geq^T_{\mathcal{A}}U\}=\{F:U\geq_\ast F\}={\mathcal{S}}$. This class extends the applicability of the ID class.

\item
The class of anti-star-shaped CDFs. In the case of absolutely continuous distributions, this is also known as the class of distribution with \textit{decreasing density on average} (DDA), as it requires the $\frac{F(x)}{x}{=\frac1x\int_0^xf(t)\,dt}$ to be decreasing. This may be obtained by $\mathcal{F}_{\mathcal{S}}^{U}=\{F:F\geq^T_{\mathcal{S}}U\}=\{F:F\geq_\ast U\}=\{ F^{-1}\in\mathcal{S}\}={\mathcal{A}}$. This is an interesting class, as it extends the applicability of the popular DD class, allowing for non-monotonicity of the PDF and jumps in the CDF.

\item
The class of distributions with a convex \textit{hazard function}, $H=-\ln (1-F)$, that is, the well-known \textit{increasing hazard rate} (IHR) class~\cite{marshall2007}, as it requires the existence of an increasing hazard rate function $h=\frac f{1-F}$ (except, possibly, at the right-endpoint of the support). This may be obtained by $\mathcal{F}_{\mathcal{V}}^{\mathcal{E}}=\{F:F\geq^T_{\mathcal{V}}\mathcal{E}\}=\{F:F^{-1}\circ \mathcal{E}\in\mathcal{V}\}=\{F:\mathcal{E}\geq_{c}F\}$. The properties and applicability of IHR models are well known.

\item
The class of distributions with a star-shaped {hazard function}. This is denoted as the \textit{IHR on average} (IHRA) class, as it requires $\frac{H(x)}{x}=\frac{1}{x}\int_0^xh(t)\,dt$, in the absolutely continuous case, to be increasing. This class may be obtained as
$\mathcal{F}_{\mathcal{A}}^{\mathcal{E}}=\{F:F\geq^T_{\mathcal{A}} \mathcal{E}\}=\{F:F^{-1}\circ \mathcal{E}\in\mathcal{A}\}=\{F:\mathcal{E}\geq_\ast F\}$. This is a relevant class (see \cite{marshall2007,shaked2007}) which extends the applicability of the IHR class (in the non-negative case).

\item
The class of distributions with a concave {hazard function}, that is, the \textit{decreasing hazard rate} (DHR) class~\cite{marshall2007}, as it requires the existence of a decreasing hazard rate function $h=\frac f{1-F}$. {Analogously, to the previous example, this class} may be obtained by $\mathcal{F}_{\mathcal{C}}^{\mathcal{E}}=\{F:F \geq^T_{\mathcal{C}}\mathcal{E}\}=\{F:F^{-1}\circ \mathcal{E}\in\mathcal{C}\}=\{F:\mathcal{E}\leq_{c}F\}$.

\item
The class of distributions with an anti-star-shaped {hazard function}. This is denoted as the \textit{DHR on average} (DHRA) class, as it requires $\frac{H(x)}{x}=\frac{1}{x}\int_0^xh(t)\,dt$,
in the absolutely continuous case, to be decreasing. This class may be obtained as
$\mathcal{F}_{\mathcal{S}}^{\mathcal{E}}=\{F:F \geq^T_{\mathcal{S}}\mathcal{E}\}=\{F:F^{-1}\circ \mathcal{E}\in\mathcal{S}\}=\{F:\mathcal{E}\leq_\ast F\}$. It extends the applicability of the DHR class (in the non-negative case).

\item
{The class of CDFs such that $\log F$ is concave, also characterised by $\frac fF$ being decreasing, known as the \textit{decreasing reversed hazard rate} (DRHR) class. This is a rather broad class of distributions.} {One may obtain this class taking} {$\mathcal{H}=\mathcal{C}$} {and} ${\mathcal{F}_{\mathcal{C}}^{\mathcal{E}_-}}=\{F:F \geq_c \mathcal{E}_-\}${, the class of functions that dominate $\mathcal{E}_-$ \wrt\ the convex transform order}.

\item
The class of distributions with a convex \textit{odds function}, $\frac{F}{1-F}$, that is, the \textit{increasing odds rate} (IOR) class~\cite{oddspaper}, as it requires the existence of an increasing odds rate function $\frac f{(1-F)^2}$ (except, possibly, at the right-endpoint of the support). This may be obtained by $\mathcal{F}_{\mathcal{V}}^{LL}=\{F:F\geq^T_{\mathcal{V}}LL\}=\{F:F^{-1}\circ LL\in\mathcal{V}\}=\{F:LL\geq_{c}F\}$. The properties and applicability of IOR models are discussed by~\cite{oddspaper} and \cite{lando2023nonparametric}.

\item
The class with a concave odds function may be similarly defined as the \textit{decreasing odds rate} (DOR) class, which may be obtained as $\mathcal{F}_{\mathcal{C}}^{LL}$.

\item
The class of distributions with a convex \textit{log-odds function}, $\log\frac{F}{1-F}$, that is, the \textit{increasing log-odds rate} (ILOR) class~\cite{zimmer1998}, as it requires the existence of an increasing log-odds rate function $\frac f{F(1-F)}$. This may be obtained by $\mathcal{F}_{\mathcal{V}}^{L}=\{F:L\geq_{c}F\}$.

\item
The class with a concave log-odds function may be similarly defined as the \textit{decreasing log-odds rate} (DLOR) class, which may be obtained as $\mathcal{F}_{\mathcal{C}}^{L}$.
\end{enumerate}

\section{Convex-ordered families}
\label{sec:conv}

In this section, we apply Theorem~\ref{general} to families of distributions which may be obtained through the convex transform order, extending some recent results of~\cite{lando2021st}. All results are summarised in the following corollaries. Although some cases are already proved in~\cite{lando2021st}, we report them here for the sake of completeness.

\begin{corollary}[\cite{lando2021st}, Corollary~3.4]\label{COR1}
If $i \geq j$, any of the following conditions imply $X_{i:n}{\ge }_{icv} X_{j:m}$.
    \begin{enumerate}
	\item
	$F$ is ID and $\frac{i}{n+1} \geq \frac{j}{m+1}$;		
	\item
	$F$ is IHR and $\sum_{k=n-i+1}^{n}{\frac{1}{k}}\geq \sum_{k=m-j+1}^{m}{\frac{1}{k}}$;
	\item
	$F$ is IOR class and $\frac{i}{n} \geq \frac{j}{m}$;
	\item
	$F$ is ILOR and $\sum_{k=i}^{n-i}\frac1k\leq\sum_{k=j}^{m-j}\frac1k$.
	\end{enumerate}
\end{corollary}
The flexibility concerning the choice of the $\mathcal{H}$ family in Theorem~\ref{general} allows for the following extension. Note that the first 4 cases of this corollary follow trivially from the previous result, and the fact that $X \leq_{icx} Y$ if and only if $-X \geq_{icv}-Y$ (Theorem~4.A.1 in \cite{shaked2007}).

\begin{corollary}\label{COR2}
If $i\leq j$, any of the following conditions imply $X_{i:n}{\ge }_{icx} X_{j:m}$.
    \begin{enumerate}
	\item
    $F$ is DD class and $\frac i{n+1}\geq \frac j{m+1}$;
	\item
    $F$ is DHR and $\sum_{k=n-i+1}^n\frac1k\geq \sum_{k=m-j+1}^m\frac1k$;
	\item
	$F$ is DOR class and $\frac{i}{n} \geq \frac{j}{m}$;
	\item
	$F$ is DLOR and $\sum_{k=i}^{n-i}\frac1k\leq\sum_{k=j}^{m-j}\frac1k$;
	\item
    If $F$ is DRHR and $\sum_{k=i}^{n}\frac 1k\leq \sum_{k=j}^{m}\frac 1k$;
	\item
    $F$ is DROR and $\frac n{n-i}\leq \frac m{m-j}$.
	\end{enumerate}
\end{corollary}
\begin{proof}
Note that if $G^{-1}\circ F$ is increasing concave, then $F^{-1}\circ G$ is increasing convex, and that the ICX order is obviously preserved under increasing convex transformations.
As follows from Theorem~4.A.63 in~\cite{shaked2007} (remark that the ICX order is, in \cite{shaked2007}, referred as 2-icx) in conjunction with Lemma~2.6 in~\cite{lando2021st}, a sufficient condition for $G^{-1}\circ B_{i:n}\geq_{icx}G^{-1}\circ B_{j:m}$ is that $i\leq j$ and $\E G^{-1}\circ B_{i:n}\geq\E G^{-1}\circ B_{j:m}$.
Then, setting $G$ as the uniform, unit exponential, LL1, standard logistic, negative exponential, and negative LL1, we obtain conditions 1--6, respectively.
We verify only case 5., the less obvious one, corresponding to $G=\mathcal{E}_{-}$, where we need to compute
$$
\E \log B_{i:n}=\int_{0}^1\frac{ t^{i-1}(1-t)^{n-i}\log t}{\mathcal{B}(i,n-i+1)}\,dt=\psi(i)-\psi(n+1)=-\sum_{k=i}^{n} \frac1k,
$$
using  repeatedly (6.44) in~\cite{Viola2016}, where $\psi(x)=\frac{\Gamma^\prime(x)}{\Gamma(x)}$, $x\geq 0$, where $\Gamma$ represents the Euler gamma function, is the \textit{digamma} function (we refer the reader to \cite{Viola2016}, for properties of $\psi$).
\end{proof}
Note that, since both the ICV and the ICX orders imply the inequality between the means, Corollaries~\ref{COR1} and \ref{COR2} provide assumptions implying that $\E X_{i:n}\geq \E X_{j:m}$. Furthermore, we may derive conditions for the comparison with the mean of their parent distribution by setting $j=m=1$ or $i=n=1$, respectively.
	
\section{Star-ordered families}
\label{sec:star}
In this section we deal with families of distribution of the form $\{F:F\geq_\ast G\}$, which include the family of anti-star-shaped CDFs and the DHRA family, using the SS order. Then, we move to families of the form $\{F:F\leq_\ast G\}$, which include the family of star-shaped distributions and the IHRA family, using the new IAS order.

\subsection{SS order of order statistics}
Let us start with some preliminary discussion. As starshapedness refers only to functions with domain $[0,+\infty)$, in this section, we will consider only non-negative RVs. First, a simple preservation property.
\begin{lemma}[\cite{shaked2007}, Theorem 4.A.56]\label{lemmass}
$X\geq_{ss}Y$ if and only if $h(X)\geq_{ss}h(Y)$, for every star-shaped function.
\end{lemma}

It is also useful to remark that a function $\phi$ is star-shaped if and only if its generalized inverse $\phi^{-1}$ is increasing anti-star-shaped, as proved in Lemma~\ref{ias-ss-1}. We now recall the following characterization of the SS order.
\begin{theorem}[\cite{shaked2007}, Theorem 4.A.54]
\label{SS4A54}
$X\geq_{ss}Y$ if and only if, for every $x\geq 0$, $\int_x^\infty t\,dF(t)\geq \int_x^\infty t\,dG(t).$
\end{theorem}

In the following subsections, we will frequently deal with transformations of beta RVs using the {result stated next}. The proof is omitted since it  follows straightforwardly, requiring a simple observation of the shape of the graphical representation of {the function considered} in each case.
\begin{lemma}
\label{ab}
Let $T_{a,b}(x)=x^a(1-x)^b$, where $a,b\in\mathbb{R}$, $c>0$, define $R(a,b,c)$ the set of roots of the equation $T_{a,b}(x)=c$ that are in $[0,1]$, and represent by $\#R(a,b,c)$ its cardinality. Then \textit{(i)} if $ab<0$, $\#R(a,b,c)=1$; \textit{(ii)} if $ab>0$, $\#R(a,b,c)\leq2$; \textit{(iii)} if $ab=0$, $\#R(a,b,c)\leq 1$.
\end{lemma}

The previous lemma means that when $c>0$, $R(a,b,c)$ has at most two elements.

Using the above lemmas, it is not difficult to apply Theorem~\ref{general} to wide families of distributions, as discussed in the next subsections.

The next theorem deals with the case of anti-star-shaped CDFs, denoted as DDA distributions.
\begin{theorem}
\label{dda}
Assume that $F$ is anti-star-shaped.
Denote by
\begin{equation}
\label{eq:SS_Z}
Z(x)=\frac{i}{n+1}\left(1-F_{B_{i+1:n+1}}(x)\right)-\frac{j}{m+1}\left(1-F_{B_{j+1:m+1}}(x)\right).
\end{equation}
If {$\frac{i}{n+1}\geq\frac{j}{m+1}$ and} for every $r\in R\left(i-j,n-i-(m-j),\frac{\mathcal{B}(i,n-i+1)}{\mathcal{B}(j,m-j+1)}\right)$, it holds that $Z(r)\geq 0$, then $X_{i:n}\geq_{ss} X_{j:m}$.
\end{theorem}
\begin{proof}
Since $F^{-1}$ is star-shaped, the result holds by Theorem~\ref{general} and Lemma~\ref{lemmass}, provided that $B_{i:n}\geq_{ss} B_{j:m}${, which, taking into account Theorem~\ref{SS4A54} and the distribution of the beta order statistics mentioned before ({$B_{i:n}\sim beta(i,n-i+1)$ and $B_{j:m}\sim beta(j,m-j+1)$} \cite{jones2004}), is equivalent to}
\begin{equation}
\label{eq:SS_star}
\int_x^1 \frac{t^{i} (1-t)^{n-i}}{\mathcal{B}(i,n-i+1)}\,dt \geq \int_x^1 \frac{t^{j} (1-t)^{m-j}}{\mathcal{B}(j,m-j+1)}\,dt,\quad\forall x\in[0,1].
\end{equation}
It is easily seen that (\ref{eq:SS_star}) is equivalent to $Z(x)\geq 0$, for every $x\in[0,1]$. {Now, the} extreme points of $Z$ are at 0, 1, or among the solutions of $T_{i-j,(n-i)-(m-j)}(x)=\frac{\mathcal{B}(i,n-i+1)}{\mathcal{B}(j,m-j+1)}$, hence the result follows immediately from the assumptions.
\end{proof}

The results of Theorem~\ref{dda} can be compared with part 1. of Corollary~\ref{COR2}. Assume that $\frac{i}{n+1}\geq\frac{j}{m+1}$ {or, equivalently, that $Z(0)\geq0$}. If $F$ is concave (DD class), then $X_{i:n}\geq_{icx}X_{j:m}$ for $i\leq j$. If $F$ is increasing anti-star-shaped (yielding the wider DDA class), then the stronger order $X_{i:n}\geq_{ss}X_{j:m}$ holds if $Z(r)\geq0,$ for $r$ in the described set. Recall that the ICX order is necessary for the SS order, and $i\leq j$ is necessary for the ICX order. So, the condition $Z(r)\geq0$, for $r$ in the set defined in Theorem~\ref{dda}, is stronger than just $i\leq j$.

We may use Theorem~\ref{dda} to get a complete geometric description of the $\geq_{ss}$-comparability of order statistics when F is DDA.
Assume the sample sizes $n\leq m$ are given. Based on Theorem~1 in~\cite{arab.tilo2021} we know that $B_{i:n}\geq_{st}B_{j:m}$, which implies $B_{i:n}\geq_{ss}B_{j:m}$, whenever $i>j$ and $n-i<m-j$, that is, whenever $i>j$.
Likewise, this result also implies that $B_{i:n}\leq_{st}B_{j:m}$, implying $B_{i:n}\leq_{ss}B_{j:m}$, whenever $i<j$ and $n-i>m-j$, which is equivalent to $n-i>m-j$ (see Figure~\ref{fig:nm}). For the region $i<j<i+m-n$ we have no $\geq_{st}$-comparability.
The line $j=\frac{m+1}{n+1}i$ corresponds to points such that $Z(0)=0$, where $Z$ is given by (\ref{eq:SS_Z}). Above this line we have $Z(0)<0$ hence, according to Theorem~\ref{SS4A54}, there is no $\geq_{ss}$-comparability. Finally, we are left with the region where $i<j$ and $\frac{i}{n+1}\geq\frac{j}{m+1}$, the region not shaded in Figure~\ref{fig:nm}, where actual verification of (\ref{eq:SS_star}) is needed.
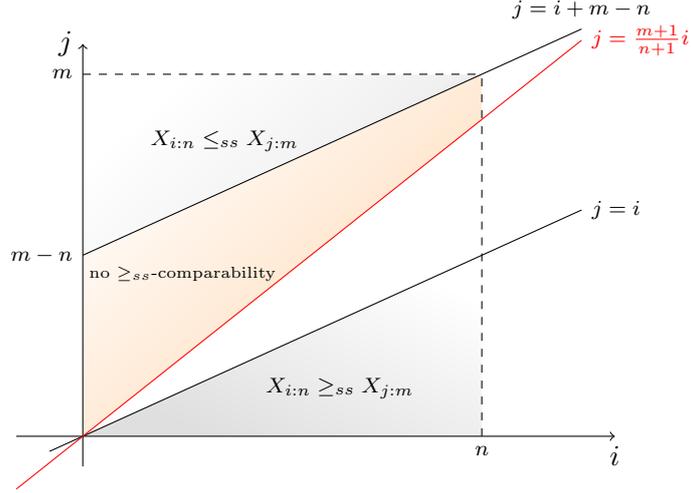
\begin{figure}
\centering
\begin{tikzpicture}[xscale=1.75,yscale=.8]
  \fill[gray!25, path fading=north,fading transform={rotate=-45}] (0,0) -- (3,0) -- (3,3) -- (0,0);
  \fill[gray!25, path fading=north,fading transform={rotate=45}] (0,3) -- (3,6) -- (0,6) -- (0,0);
  \fill[orange!30, path fading=north,fading transform={rotate=45}] (0,0) -- (3,5.25) -- (3,6) -- (0,3) -- (0,0);
  \draw[->] (-.5,0) -- (4,0) coordinate[label = {below:$i$}] (xmax);
  \draw[->] (0,-.5) -- (0,6.5) coordinate[label = {left:$j$}] (ymax);
  \draw (-.25,-.25) -- (3.75,3.75) coordinate[label = {right:{\scriptsize $j=i$}}];
  \draw[red] (-.5,-.875) 
      -- (3.75,6.563) coordinate[label = {right:{\scriptsize $j=\frac{m+1}{n+1}i$}}];
  \draw (0,3) coordinate[label = {left:{\scriptsize $m-n$}}] -- (3.75,6.75) coordinate[label = {above:{\scriptsize $j=i+m-n$}}];
  \draw[dashed] (0,6) coordinate[label = {left:{\scriptsize $m$}}] -- (3,6);
  \draw[dashed] (3,0) coordinate[label = {below:{\scriptsize $n$}}] -- (3,6);
  \draw (1.3,.8) coordinate[label={right:{\scriptsize $X_{i:n}\geq_{ss}X_{j:m}$}}];
  \draw (1.7,4.9) coordinate[label={left:{\scriptsize $X_{i:n}\leq_{ss}X_{j:m}$}}];
  \draw (.75,2.7) coordinate[label={center:{\tiny no $\geq_{ss}$-comparability}}];
\end{tikzpicture}
\caption{$\geq_{ss}$-comparability for distributions in the DDA class.}\label{fig:nm}
\end{figure}
{For $(i,j)$ in the unshaded region it is easily seen that $Z^\prime(x)=-\frac{x^i(1-x)^{n-i}}{\mathcal{B}(i+1,n-i+1)}+\frac{x^j(1-x)^{m-j}}{\mathcal{B}(j+1,m-j+1)}<0$ whenever $x$ is close to 0 or 1. Moreover, as Lemma~\ref{ab} implies that $Z$ has two extreme points in $(0,1)$, the monotonicity of $Z$ is ``$\searrow\nearrow\searrow$''.
A numerical verification shows that the initial interval where $Z$ is decreasing is rather small, so $Z$ will remain nonnegative whenever $Z(0)=\frac{i}{n+1}-\frac{j}{m+1}>0$ is large enough. Therefore, we expect that points $(i,j)$ not satisfying the assumption in Theorem~\ref{dda} will be close to the top border of the unshaded region. A few examples illustrating this behaviour are shown in Figure~\ref{fig1}.}
\begin{figure}[h]
\centering
\includegraphics[scale=.45]{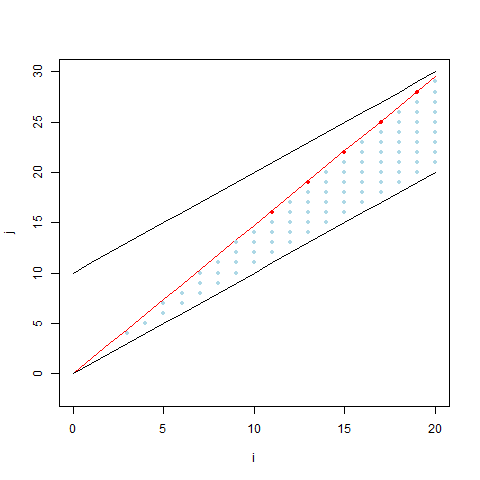}\hfil \includegraphics[scale=.45]{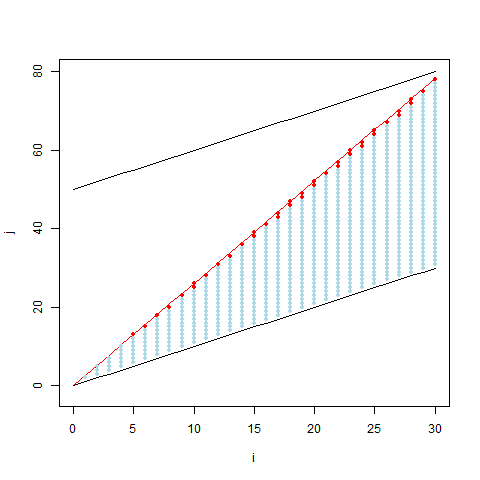}
\caption{Blue points fulfill the assumptions of Theorem~\ref{dda}.  Left: $n=20$, $m=30$. Right: $n=30$, $m=80$.}
\label{fig1}
\end{figure}

We now extend our approach to the family of DHRA distributions.

\begin{theorem}
Assume that $F$ is DHRA. Let
\begin{multline}
Z(x)={n\choose i}{i}\sum_{k=0}^{i-1} {{i-1}\choose k}(-1)^{i-1-k}\frac{\left(e^{-x}\right)^{n-k} (-k x+n x+1)}{(k-n)^2}\\
		- {m\choose j}{j}\sum_{k=0}^{j-1} {{j-1}\choose k}(-1)^{j-1-k}\frac{\left(e^{-x}\right)^{m-k} (-k x+m x+1)}{(k-m)^2}.
\end{multline}
If $Z(0)\geq0$ and $Z(-\log(1-r))\geq 0$ for every $r\in R\left(i-j,(n-i)-(m-j),\frac{\mathcal{B}(i,n-i+1)}{\mathcal{B}(j,m-j+1)}\right)$ then $X_{i:n}\geq_{ss} X_{j:m}$.
\end{theorem}
\begin{proof}
Since $F^{-1}\circ \mathcal{E}$ is star-shaped, the result holds by Theorem~\ref{general} and Lemma~\ref{lemmass}, provided that $\mathcal{E}^{-1}\circ B_{i:n}\geq_{ss} \mathcal{E}^{-1}\circ B_{j:m}${, or, equivalently, $-\log(1-B_{i:n})\geq_{ss}-\log(1-B_{j:m})$}. This may be expressed as
\begin{multline}
\frac{1}{\mathcal{B}(i,n-i+1)}\int_x^\infty t\left(1-e^{-t}\right)^{i-1} {e^{-(n-i+1)t}}\,dt \geq\\
		\frac1{\mathcal{B}(j,m-j+1)}\int_x^\infty t \left(1-e^{-t}\right)^{j-1} {e^{-(m-j+1)t}}\,dt,\quad\forall x\geq0,
\label{Z.DHRA}
\end{multline}
Using the binomial Theorem, we obtain
\begin{multline}
\frac{1}{\mathcal{B}(i,n-i+1)}\int_x^\infty t\left(1-e^{-t}\right)^{i-1} {e^{-(n-i+1)t}}\,dt\\
		=\frac{1}{\mathcal{B}(i,n-i+1)}\sum_{k=0}^{i-1}\int_x^\infty {(-1)^{i-1-k}} {{i-1}\choose k}t{e^{-(n-k)t}}\,dt\\
		={n\choose i}{i}\sum_{k=0}^{i-1} {{i-1}\choose k}(-1)^{i-1-k}\frac{{e^{(n-k)t}} (-k x+n x+1)}{(k-n)^2},
\end{multline}
and similarly for the second term, {hence (\ref{Z.DHRA}) is equivalent to} $Z(x)\geq 0$, for every $x\geq0$. Note that $Z(+\infty)=\lim_{x\to+\infty}Z(x)=0$, due to the exponential terms. Now, the function $Z\circ \mathcal{E}^{-1}$ is continuous on $[0,1]$, so it is nonnegative if and only if
its minimal value in $[0,1]$ is nonnegative.
The extreme points of {$Z\circ \mathcal{E}^{-1}$} are easily seen to be among the solutions of $T_{i-j,(n-i)-(m-j)}(1-e^{-x})=\frac{\mathcal{B}(i,n-i+1)}{\mathcal{B}(j,m-j+1)}$, hence the result follows immediately from the assumption ${Z\circ\mathcal{E}^{-1}(r)}\geq 0$ for every $r\in R\left(i-j,(n-i)-(m-j),\frac{\mathcal{B}(i,n-i+1)}{\mathcal{B}(j,m-j+1)}\right)$.
\end{proof}

A complete geometric picture of the $\geq_{ss}$-comparability for DHRA distributions produces a plot similar to the one in Figure~\ref{fig:nm}. The shaded regions where one has comparability are the same, but the directions of the $\geq_{ss}$-comparability are reversed, taking into account that $1-B_{i:n}$ and $1-B_{j:m}$ still have $beta$ distributions with the parameters swapped.
Moreover, the red line in Figure~\ref{fig:nm} is now replaced by setting to 0 the two terms appearing in part 2. of Corollary~\ref{COR2}, that is, for each $i\leq n$ going through the coordinates $j$ and $j+1$ such that $\sum_{k=n-i+1}^n\frac1k-\sum_{k=m-j+1}^m\frac1k$ and $\sum_{k=n-i+1}^n\frac1k-\sum_{k=m-(j+1)+1}^m\frac1k$ have opposite signs.
The region below this curve, corresponding to $Z(0)>0$, and above the diagonal is seen to be where we have no $\geq_{ss}$-comparability. The remaining region needs numerical verification. Hence, with respect to Figure~\ref{fig:nm}, one reverses the direction of the comparisons, swaps the unshaded and shaded areas between the two straight lines, and, the separating red line is no longer straight.

\subsection{Properties of the IAS order}
\label{sec:ias.ss}

Our method can be applied to classes of the form $\{F:G\geq_* F\}$ using the IAS order. This includes the important IHRA class, obtained for $G=\mathcal{E}$, and also the IDA class, where we take $G=U$, the uniform distribution.
In some sense the IAS order behaves like the SS order, with the disadvantage that it does not seem to have a simple characterization based on a transformation of the CDFs, analogous to Theorem~\ref{SS4A54}, which makes it difficult to check.

The IAS order satisfies the following properties.
\begin{proposition}
\label{prop:ias}
Let $X$ and $Y$ be nonnegative random variables with CDFs $F$ and $G$, respectively.
    \begin{enumerate}
    \item
    $X\geq_{st}Y\quad\Rightarrow\quad X\geq_{ias}Y\quad\Rightarrow\quad X\geq_{icv}Y$.
    \item
    $X\geq_{ias}Y$ implies $h(X)\geq_{ias}h(Y)$, for every increasing anti-star-shaped~$h$.
    \item
    Let $\Theta$ be a random variable, and let $F(\cdot |\theta)$ and $G(\cdot |\theta)$ be the conditional CDFs of $X$ and $Y$ with respect to the event $\Theta=\theta$. If $F(\cdot|\theta)\geq_{ias}G(\cdot|\theta)$ for every possible realization $\theta$ of $\Theta$, then $X\geq_{ias}Y$.
	\end{enumerate}
\end{proposition}
\begin{proof}
\begin{enumerate}
\item
The first implication follows from the fact that all increasing anti star-shaped functions are increasing. Let $\phi$ be an increasing concave function. If $\phi(0)=0,$ $\phi$ is also anti-star-shaped and {$X\geq_{ias}Y$ implies that $\E\phi(X)\geq\E\phi(Y)$}. If $\phi(0)\neq0$, define $\phi_0(x)=\phi(x)-\phi(0)$, which is still increasing concave, and proceed similarly.
\item
Let $\phi$ be any increasing anti-star-shaped function. Taking into account that $h$ is increasing, the quotient
$$
\frac{\phi(h(x))}{x} = \frac{\phi(h(x))}{h(x)} \frac{h(x)}{x},
$$
is the product of two decreasing functions, hence it is decreasing itself. That is, $\phi\circ h$ is increasing anti-star-shaped. Therefore, $X\geq_{ias} Y$ implies that $\E(\phi \circ h(X))\geq \E(\phi \circ h (Y))$ or, equivalently, $h(X)\geq_{ias}h(Y)$.
\item This follows directly from the tower law of conditional expectations.
\end{enumerate}
\end{proof}	
Properties 2. and 3. mean that the IAS order is closed under composition and mixtures, respectively. Property 1. shows why the IAS order can be useful. In fact, it measures size and dispersion at the same time. The IAS order implies the inequality of the means, moreover, if $\E X=\E Y$, then $X\geq_{ias}Y$ implies that ${\rm Var}(X)\leq{\rm Var}(Y)$. {Below, we provide examples showing that neither of the implications in part 1. of Proposition~\ref{prop:ias} is an equivalence.} The above properties suggests that the IAS order can be used as a valid (and stronger) alternative to the commonly-used ICV order, whenever we deal with star-ordered families. However, for technical reasons, the verification of the IAS order is complicated, as discussed in the next subsection.

{\begin{example}
\label{ex:ias-notst}
Let $p\in(0,1)$, $p^\prime>p$, and $s\in (0,1)$ be given, and consider the random variables $X$ and $Y$ with distributions $P(X=p)=1-P(X=0)=s$, and $P(Y=p^\prime)=1-P(Y=0)=r=\frac{sp}{p^\prime}$. As $r<s$, the CDFs of $X$ and $Y$ cross, then $X\not\geq_{st}Y$. However, given any increasing anti-star-shaped function $\phi$, so that $\phi(0)=0$ and $\frac{\phi(x)}{x}$ is decreasing, we have $\E\phi(X)=s\phi(p)=sp\frac{\phi(p)}{p}\geq sp\frac{\phi(p^\prime)}{p^\prime}=r\phi(p^\prime)=\E\phi(Y)$. Hence $X\geq_{ias}Y$, showing that the first implication in part 1. of Proposition~\ref{prop:ias} is indeed not an equivalence.
\end{example}}

{\begin{example}
\label{ex:icv-notias}
To show that also the second implication is not an equivalence, take $X$ exponentially distributed, with CDF $\mathcal{E}$, and $Y$ with Weibull distribution, with shape and scale parameters equal to 2 and $\Gamma(\frac32)$, respectively. Therefore, $\E X=\E Y$ and their CDFs cross once. According to Theorem~4.A.22 in \cite{shaked2007}, $Y\geq_{icv}X$. On the other hand, considering the increasing anti-star-shaped function
$$
\phi(x)=\begin{cases}
    		10x, & x<\frac{1}{10},\\
            1, &  \frac{1}{10}\leq x<1, \\
    		x, & x\geq 1,\\
    	\end{cases}
$$
it is easy to verify that $\E\phi(X)>\E\phi(Y)$, so $Y\not\geq_{ias}X$.
\end{example}}

\subsection{IAS order of order statistics: a heuristic approach}
\label{subsec:ias}
From a practical point of view, a simple characterization of the IAS order, described in distributions terms, seems unavailable and remains an open problem. An alternative approach may be based on Theorem~\ref{thm:ias} below. To state our result, we need some additional notations. Given $a>0,$ define
$$
u_a(x)=	\begin{cases}
    		0, &x<a,\\
    		x, &x\geq a.\\
    	\end{cases} $$
Moreover, given sequences $a_1,\ldots,a_n>0$ and $b_1,\ldots,b_n>0$, define
\begin{equation}
\label{eq:s_n}
s_n(x)=\sum_{\ell=1}^n b_\ell u_{a_\ell}(x).
\end{equation}

\begin{theorem}
\label{thm:ias}
Let $X\sim F$ and $Y\sim G$ be absolutely continuous nonnegative random variables with density functions $f$ and $g$, respectively. $X\geq_{ias}Y$ if and only if, for every positive integer $n$, and for every $0=a_0\leq a_1\leq\cdots\leq a_n$ and $b_1,\ldots,b_n>0$, and every $s_n$ defined according to (\ref{eq:s_n}), we have:
\begin{equation}
\label{eq:ias}
\begin{array}{l}
\displaystyle\int_0^\infty s_n(t)f\circ s_n(t)\,dt+\sum_{\ell=1}^n a_\ell(F(a_\ell)-F(a_{\ell-1})) \\
\displaystyle\qquad\qquad\geq \int_0^\infty s_n(t)g\circ s_n(t)\,dt+\sum_{\ell=1}^n a_\ell(G(a_\ell)-G(a_{\ell-1})).
\end{array}
\end{equation}
\end{theorem}
\begin{proof}
As proved in Lemma~\ref{ias-ss-1}, $h$ is increasing anti-star-shaped if and only if $h=s^{-1}$, where $s$ is star-shaped. By the change of variable $h(x)=t$, $\int_0^\infty h(x)\,dF(x)=\int_0^\infty t\,dF\circ s(t)$, and equivalently for the integration with respect to $G$. Hence, we need to prove that $\int_0^\infty t\,dF\circ s(t)\geq \int_0^\infty t\,dG\circ s(t)$ for every star-shaped function $s$. Now, every star-shaped function $s$ can be approximated by a sequence $s_n$, where the sequences $a_\ell$ and $b_\ell$ satisfy the given assumptions. Indeed, this follows directly from approximating the increasing function $\frac{s(x)}{x}$ by an increasing step-function.
Therefore, by monotonous approximation,	$X\geq_{ias}Y$ if and only if $\int_0^\infty t\,dF\circ s_n(t)\geq \int_0^\infty t\,dG\circ s_n(t)$ for every integer $n$, and every $a_\ell$ and $b_\ell$ as given. As $F\circ s_n$ and $G\circ S_n$ have discontinuities at $a_1,\ldots,a_n$, and $F$ and $G$ have densities $f$ and $g$, respectively, we obtain
$$
\int_0^\infty t\,dF\circ s_n(t)=\int_0^\infty s_n(t)f\circ s_n(t)\,dt+\sum_{\ell=1}^n a_\ell(F(a_\ell)-F(a_{\ell-1})),
$$
and similarly for $\int_0^\infty t\,dG\circ s_n(t)$.
\end{proof}
    		
Although Theorem~\ref{thm:ias} provides a necessary and sufficient condition, it requires the verification of infinitely many inequalities, thus reducing its usability for a direct verification of the IAS order. To address this difficulty we propose the following simulation algorithm to check if $X\geq_{ias}Y$. Let $K$ be the total number of repetitions; for every $k=1,\ldots,K$:
\begin{enumerate}
\item
randomly generate $n$, say $n(k)$, from a discrete distribution with infinite support;
\item
randomly generate the sequences $0=a_0\leq a_1\leq\cdots\leq a_{n(k)}$ and $b_1,\ldots,b_{n(k)}$ and define $s_n$ according to (\ref{eq:s_n});
\item
compute
$$
I_F(k)=\int_0^\infty s_n(t)f\circ s_n(t)\,dt+\sum_{\ell=1}^n a_\ell(F(a_\ell)-F(a_{\ell-1})),
$$
and
$$
I_G(k)=\int_0^\infty s_n(t)g\circ s_n(t)\,dt+\sum_{\ell=1}^n a_\ell(G(a_\ell)-g(a_{\ell-1}));
$$
\item
if $R(K)=\frac1K \sum_{k=1}^K\mathbb{I}(I_F(k)-I_G(k))<1$ we have that $X\not\geq_{ias}Y$, otherwise, if $R(K)=1$ we have an indication that $X$ may dominate $Y$ in the IAS order.
\end{enumerate}
    		
The procedure described above can be used to check whether $X_{i:n}\geq_{ias}X_{j:m}$.
    			
\begin{corollary}
\label{cor:ias1}
Assume that $G\geq_\ast F$. If for every positive integer $n$, and for every choice of $0=a_0\leq a_1\leq\cdots\leq a_n\leq 1$ and $b_1,\ldots,b_n>0$,
\begin{equation}
\label{eq:ias1}
\begin{array}{l}
\displaystyle\int_0^\infty s_n(t)f_{B_{i:n}}\circ G^{-1}\circ s_n(t)\,dt+\sum_{\ell=1}^n a_\ell\left(F_{B_{i:n}}\circ G^{-1}(a_\ell)-F_{B_{i:n}}\circ G^{-1}(a_{\ell-1})\right) \\
\displaystyle\qquad\geq \int_0^\infty s_n(t)f_{B_{j:m}}\circ G^{-1}\circ s_n(t)\,dt+\sum_{\ell=1}^n a_\ell\left(F_{B_{j:m}}\circ G^{-1}(a_\ell)-F_{B_{j_m}}\circ G^{-1}(a_{\ell-1})\right),
\end{array}
\end{equation}
then $X_{i:n}\geq_{ias}X_{j:m}$.
\end{corollary}
    			
\begin{proof}
According to Theorem~\ref{thm:ias}, (\ref{eq:ias1}) is equivalent to $G^{-1}\circ B_{i:n}\geq_{ias} G^{-1}\circ B_{j:m}$. Then, the result follows from the fact that  $F^{-1}\circ G $ is increasing anti-star-shaped and this class is closed under composition, noting that, as we are integrating with respect to a beta distribution, we only need to consider the approximation in [0,1].
\end{proof}

Similarly to the previous applications, choosing a particular $G$ leads to conditions for $X_{i:n}\geq_{ias}X_{j:m}$, when $F$ belongs to the appropriate family of distributions. For example, taking $G=U$ we find conditions that apply when $F$ is star-shaped (or in the IDA class, referring to the families described in Section~\ref{sec:classes}), while the choice $G=\mathcal{E}$ gives conditions when $F$ is IHRA.

\begin{example}
\label{ex:example}
Take $n = 200$, $i =70$, $m = 70$, $j=8$. In this case we have $B_{70:200}\not\geq_{st} B_{8:70},$ which means that $X_{70:200}\not\geq_{st} X_{8:70}$. Hence, as the strongest of the stochastic order fails to hold, we may be interested in checking that some weaker order, such as ICV or IAS, holds. If $F$ is convex, these values satisfy $i\geq j$ and $\frac{i}{n+1}\geq\frac{j}{m+1}$, then $X_{70:200}\geq_{icv}X_{8:70}$.
However, if $F$ is not convex but only star-shaped, we can check the condition of Corollary~\ref{cor:ias1} for $G=U$, using the proposed algorithm. Taking $K=1000$ and random generating $n$ from a Poisson distribution with parameter $\lambda=20$, we obtain $R(1000)=1$, suggesting, although not actually proving, that $X_{70:200}\geq_{ias}X_{8:70}$. This would imply that $\E X_{70:200}\geq\E X_{8:70}$, although $X_{70:200}\not\geq_{st}X_{8:70}$, as the identity is an increasing anti-star-shaped function. Note that $X_{70:200}\geq_{icv}X_{8:70}$ means that $\E u(X)\geq\E u(Y)$ for every increasing concave function $u$, a large subset of the class of increasing anti-star-shaped functions. Nevertheless, the proposed algorithm was not able to identify an increasing anti-star-shaped function violating (\ref{eq:ias1}). Although this, of course, it is not guaranteed that such a function does not exist.
Now, take $n= 200$, $i = 65$, $m = 40$, $j= 10$. Even in this case, $X_{65:200}\not\geq_{st} X_{10:40}$. If $F$ is convex we still have $X_{65:200}\geq_{icv}X_{10:40}$. However, applying again our algorithm (with the same settings), we obtain $R(1000)=0.89$, so in this case we know that $B_{65:200}\not\geq_{ias}B_{10:40}$. Accordingly, one cannot conclude that $X_{65:200}$ dominates $X_{10:40}$ in the IAS order, although we assume that $F$ is star-shaped. Similar examples can be provided for the IHRA case.
\end{example}

\section{Bounds for probabilities of exceedance}
\label{sec:bounds}

Consider a scenario where we represent the lifetime of a $k$-out-of-$n$ system as $X_{k:n}$. A notable challenge in reliability analysis involves determining the probability that the individual component's lifetime falls below or exceeds the expected lifetime of the entire system, denoted as {$\mathbb{E}X_{k:n}$. In a parametric setting, this probability can be precisely computed using the mathematical formula of the parent CDF $F$. However, when the exact form of $F$ is unknown, we can leverage information about its overall shape to establish upper or lower bounds for this probability.  This would follow from the application of Jensen's inequality, under the assumption that $F$ belongs to a convex-ordered family $\mathcal{F}_{\mathcal{V}}^G$ or $\mathcal{F}_{\mathcal{C}}^G$. We remark that the case in which $G$ is the uniform has been already discussed by \cite{ali1965}.

\begin{proposition}
Given a CDF $G$, define $p_{i:n}^G=G(\E(G^{-1}\circ B_{i:n}))$. For every $i$ and $n$ such that $\E (G^{-1}\circ B_{i:n})$ is defined, the following holds.
\begin{enumerate}
\item 	
Let $F\in\mathcal{F}_{\mathcal{V}}^G.$ Then, $P(X\leq \E X_{i:n})\leq p_{i:n}^G$.
\item 	
Let $F\in\mathcal{F}_{\mathcal{C}}^G.$ Then, $P(X\leq \E X_{i:n})\geq p_{i:n}^G$.
\end{enumerate}
In particular, given a pair of CDFs $G_1$ and $G_2$, if $G_1\leq_{st} G_2$, then $p_{i:n}^{G_1}\geq p_{i:n}^{G_2}$.
\end{proposition}
\begin{proof}
We prove just case 1., as case 2. is dealt similarly. We are assuming that $F^{-1}\circ G$ is concave, hence $G^{-1}\circ F$ is convex. Therefore,  Jensen's inequality gives $\E X_{i:n}\leq F^{-1}\circ G(\E(G^{-1}\circ F(X_{i:n})))$.
Now, taking into account that $F(X_{i:n})\sim B_{i:n},$ applying $F$ to both sides we obtain
$P(X\leq \E X_{i:n})\leq G(\E(G^{-1}B_{i:n}))$.
Note that $G_1\leq_{st} G_2$ implies that $G_1^{-1}(B_{i:n})\leq_{st} G_2^{-1}(B_{i:n})$, so that $\E(G_1^{-1}B_{i:n})\leq\E(G_2^{-1}B_{i:n})$, hence the last statement follows.
\end{proof}

Put otherwise, the above result means that, if $F\in\mathcal{F}_{\mathcal{V}}^G$, the expected order statistic $\E X_{i:n}$ is always smaller than or equal to the  $ p_{i:n}^G$-quantile of $X$. Similarly, if $F\in\mathcal{F}_{\mathcal{C}}^G$, the expected order statistic $\E X_{i:n}$ is always greater than or equal to the  $ p_{i:n}^G$-quantile of $X$, that is, $\E X_{i:n}\geq F^{-1}(p_{i:n}^G)$. This result also enables a useful characterization of the LL1 distribution. Indeed, generally one may approximate $\E X_{i:n}$ with $F^{-1}(\frac{i}n)$: for $n\to\infty$ and $\frac{i}n\to p$ (constant), $F^{-1}(\frac{i}n)\to\E X_{i:n}$. This result is exact for $n$ finite if and only if $X$ has an LL1 distribution.
\begin{corollary}[A characterisation of the LL1 distribution]
$\E X_{i:n}=F^{-1}(\frac{i}{n})$ if and only if $F(x)=LL(\frac x a)$, for any scale parameter $a>0$.
\end{corollary}
\begin{proof}
First, note that, for the LL1 distribution, the expectations $\E X_{i:n}$ are finite for every $i=1,\ldots,n-1,$ while $\E X_{n:n}=F^{-1}(1)=\infty$.
$F$ belongs to both $\mathcal{F}_{\mathcal{C}}^{LL}$ and $\mathcal{F}_{\mathcal{V}}^{LL}$ if and only if $F(x)= LL(\frac x a)$. Without loss of generality, let $a=1$. In this case, it is easy to verify that $p_{i:n}^{LL}=\frac{i}n$, so $\frac{i}n\leq LL(\E X_{i:n})\leq\frac{i}n$. This means that the $\frac{i}n$-quantile of the LL1 is $\E(X_{i:n})=\frac{i}{n-i}$.
\end{proof}
Common choices of $G$ yield the following explicit expressions of $p_{i:n}^G$:
\begin{enumerate}
	\item If $G=U,$ $p_{i:n}^U=\frac{i}{n+1}$.
	\item If $G=\mathcal{E}$, $p_{i:n}^\mathcal{E}=1-\exp\left(-\sum_{k=n-i+1}^n \frac1k\right)$.
	\item If $G=LL$,  $p_{i:n}^{LL}=\frac{i}n$.
\item	If $G=\mathcal{E_-}$, $p_{i:n}^\mathcal{E_-}=\exp\left(-\sum_{k=i}^n \frac1k\right)$.
\end{enumerate}

Table~\ref{t1} shows the $p_{i:n}^G$ bounds for $n=10$ and some choices of $G$.
The application of our results is quite straightforward. For instance, if we know that the CDF of interest, $F$, is IHR and has a decreasing density, as is the case, for example, of the Gompertz distributions (for suitably chosen parameters), then the probability of having $X\leq \E X_{i:n}$ is always between $p_{i:n}^{U}$ and $p_{i:n}^{\mathcal{E}}$, that is,
$$
\frac{i}{n+1}\leq P(X\leq \E X_{i:n})\leq 1-\exp\left(-\sum_{k=n-i+1}^n \frac1k\right).
$$
If $i=3$ and $n=10$, this means that $P(X\leq \E X_{i:n})\in{[0.273,0.285]}$.
Similarly, if $F$ is IOR and DRHR, then
$$
\exp\left(-\sum_{k=i}^n \frac1k\right)\leq P(X\leq \E X_{i:n})\leq \frac{i}{n}.
$$
As these classes are wider than the previously considered, these bounds are generally weaker, so, for $i=3$ and $n=10$, we now find $P(X\leq \E X_{i:n})\in{[0.240,0.300]}$.
\begin{table}[h]
\centering
\begin{tabular}{|c|rrrrrrrrrr|}
	\hline
$G$	& $i=1$& $i=2$ & $i=3$ & $i=4$ & $i=5$ & $i=6$ & $i=7$ & $i=8$ & $i=9$ & $i=10$ \\
	\hline
	$LL$	& 0.100 & 0.200 & 0.300 & 0.400 & 0.500 & 0.600 & 0.700 & 0.800 & 0.900 & 1.000 \\
	$\mathcal{E}$ & 0.095 & 0.190 & 0.285 & 0.381 & 0.476 & 0.571 & 0.666 & 0.760 & 0.855 & 0.947 \\
	$U$ & 0.091 & 0.182 & 0.273 & 0.364 & 0.455 & 0.545 & 0.636 & 0.727 & 0.818 & 0.909 \\
	$\mathcal{E}_-$ & 0.053 & 0.145 & 0.240 & 0.334 & 0.429 & 0.524 & 0.619 & 0.715 & 0.810 & 0.905 \\
	\hline
\end{tabular}
\caption{$p_{i:10}^G$ for different choices of $G$.}\label{t1}
\end{table}
The bounds, with respect to the families of distributions, are sharp, as illustrated in Figure~\ref{fig:bounds}, where we plotted true probabilities for two distributions that are both IRH and DRHR (Weibull with shape parameter larger than 1, and power distribution), and the inverted power distribution with exponent $\frac12$, which is not IHR, hence violates the lower bound.
\begin{figure}[h]
\centering
\begin{tikzpicture}[xscale=5,yscale=70]
\renewcommand{\TickSize}{.15pt}%
\draw (-.05,0) node [left]{{\scriptsize $\exp\left(-\sum_{k=i}^n \frac1k\right)$}} -- (1.2,0);
\foreach \x in {.1,.3,.5,.7,.9} {%
    \draw ($(\x,0) + (0,-\TickSize/5)$) node [below] {{\scriptsize $\x$}} -- ($(\x,0) + (0,\TickSize/5)$);}
\draw[->] (0,-.031) -- (0,.048);
\foreach \y in {-.03,-.015,.015,.03,.045} {%
    \draw ($(0,\y) + (-\TickSize*3,0)$) node [left] {{\scriptsize $\y$}} -- ($(0,\y) + (\TickSize*3,0)$);}
\draw (.1,0.04171042) -- (.2,0.04501787) -- (.3,0.04590034) --(.4,0.04625058) -- (.5,0.04638576) --
      (.6,0.04638576) -- (.7,0.04625058) -- (.8,0.04590034) -- (.9,0.04501787) -- (1,0.04171042);
\draw[loosely dashed,red] (.1,0.01527902) -- (.2,0.01785149) -- (.3,0.01968728) --(.4,0.02146076) --
                          (.5,0.02327178) -- (.6,0.02515650) -- (.7,0.02713690) -- (.8,0.02921840) --
                          (.9,0.03131254) -- (1,0.03241204);
\draw[dashed,blue] (.1,0.013282984) -- (.2,0.012888947) -- (.3,0.011639142) -- (.4,0.010247637) --
                   (.5,0.008811972) -- (.6,0.007357585) -- (.7,0.005893680) -- (.8,0.004424324) --
                   (.9,0.002951572) -- (1,0.001476569);
\draw[dashdotted,BlueViolet] (.1,0.033676912) -- (.2,0.028979144) -- (.3,0.021895098) -- (.4,0.014333784) --
                             (.5,0.006639215) -- (.6,-0.001060823) -- (.7,-0.008671621) -- (.8,-0.016055052) --
                             (.9,-0.022884813) -- (1,-0.027928909);
\end{tikzpicture}
\caption{Upper bounds (black) and true values for $P(X\leq \E X_{i:n})$: Weibull(3,1) (red), $F(x)=x^3$ (blue), $F(x)=1-\sqrt{1-x}$ (dashdotted, dark blue) with respect to the lower bounds (the horizontal line).}
\label{fig:bounds}
\end{figure}
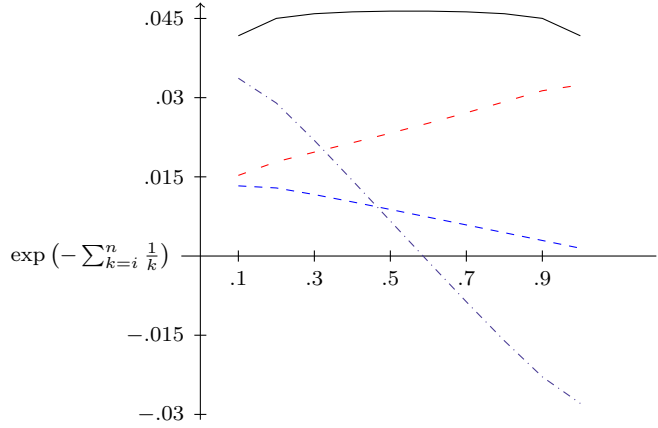

\subsection{Application}
The paper~\cite{nichols2006} provides a table containing a sample of size $n=100$ of breaking stress for carbon fibers. Applying the tests of~\cite{landoihr} and~\cite{lando2023nonparametric}, respectively, it can be tested that this dataset is likely to come from an IOR distribution. Moreover, the dataset also seems to satisfy the DRHR assumption.
A straightforward interval for $\E X_{i:n}$ is then obtained by plugging in the appropriate bounds described above to the empirical CDF, that is $\left[\F_n^{-1}\left(\exp\left(-\sum_{k=i}^n \frac1k\right)\right), \F_n^{-1}(\frac{i}n)\right]$. Taking, as an example, $i=20$, this interval reduces to a single point, as we get $1.69\leq\widehat{\E} X_{20:100}\leq 1.69$. However, we may use instead estimators that take into account the available information about the shape of the distribution: we may use $\F_n^{IOR}$, introduced by~\cite{lando2023nonparametric}, as an IOR estimator of the CDF, and $\F_n^{DRHR}$, proposed by~\cite{sengupta2005}, as a DRHR estimator. Differently from the empirical CDF, these estimators are continuous. Hence, an interval for $\E X_{i:n}$ may be given by $\left[(\F_n^{DRHR})^{-1}\left(\exp\left(-\sum_{k=i}^n \frac1k\right)\right), (\F_n^{IOR})^{-1}(\frac{i}{n})\right]$. For this sample, thus taking into account the knowledge about the shape of the CDF, this leads to $\E X_{20:100}\in[1.623, 1.716]$.


\section*{Funding}
T.L. was supported by the Italian funds ex MURST 60\% 2022. I.A. and P.E.O. were partially supported by the Centre for Mathematics of the University of Coimbra UID/MAT/00324/2020, funded by the Portuguese Government through FCT/MCTES and co-funded by the European Regional Development Fund through the Partnership Agreement PT2020.

\bibliographystyle{plainnat} 
\bibliography{biblio}

\end{document}